\documentclass[conference,a4paper]{IEEEtran}

\newcommand{\bs}[1]{\boldsymbol{#1}}

\newcommand{\bk}[1]{\left(#1\right)}
\newcommand{\Bk}[1]{\left[#1\right]}
\newcommand{\BK}[1]{\left\{#1\right\}}

\newcommand{\trace}{\operatorname{tr}}
\newcommand{\expect}{\operatorname{\mathbb E}}

\usepackage{amsmath}
\usepackage{amssymb}
\usepackage{graphicx}
\usepackage{amsfonts}
\usepackage{graphicx}
\usepackage{amsthm}

\newtheorem{lemma}{Lemma}
\newtheorem{theorem}{Theorem}

\newtheorem{cor}{Corollary}
\theoremstyle{remark}
\newtheorem*{remark}{Remark}

\begin{document}

\sloppy

\title{Mismatched Quantum Filtering and Entropic Information} 

\author{
  \IEEEauthorblockN{Mankei Tsang\IEEEauthorrefmark{1}\IEEEauthorrefmark{2}}
  \IEEEauthorblockA{\IEEEauthorrefmark{1}
    Department of Electrical \& Computer Engineering, 
National University of Singapore\\
4 Engineering Drive 3, Singapore 117583} 
  \IEEEauthorblockA{\IEEEauthorrefmark{2}
Department of Physics, National University of Singapore\\
2 Science Drive 3, Singapore 117551\\
    Email: eletmk@nus.edu.sg} 
}



\maketitle

\begin{abstract}
  Quantum filtering is a signal processing technique that estimates
  the posterior state of a quantum system under continuous
  measurements and has become a standard tool in quantum information
  processing, with applications in quantum state preparation, quantum
  metrology, and quantum control. If the filter assumes a nominal
  model that differs from reality, however, the estimation accuracy is
  bound to suffer.  Here I derive identities that relate the excess
  error caused by quantum filter mismatch to the relative entropy
  between the true and nominal observation probability measures, with
  one identity for Gaussian measurements, such as optical homodyne
  detection, and another for Poissonian measurements, such as photon
  counting. These identities generalize recent seminal results in
  classical information theory and provide new operational meanings to
  relative entropy, mutual information, and channel capacity in the
  context of quantum experiments.
\end{abstract}

\section{Introduction}
Long regarded as an afterthought in the development of quantum
mechanics, the probabilistic nature of quantum measurements is now
taking the center stage in theoretical and experimental physics
\cite{haroche_nobel,wineland_nobel}. Quantum probability theory will
inevitably play a more prominent role in not just fundamental science
but also future technology, which will require increasingly precise
estimation and control of physical devices in the quantum regime.

Most of the current quantum information processing technology relies
on continuous electromagnetic fields to measure and control quantum
devices. The Bayesian quantum filtering theory, pioneered by Belavkin
\cite{belavkin}, enables one to estimate the state of a quantum system
from a continuous field measurement record and has therefore become a
standard tool in the area. The theory is applicable to a wide range of
current experiments, including those on atoms, mechanical oscillators,
or superconducting circuits interacting with optical or microwave
fields \cite{wiseman_milburn}. Foreseeable applications include, but
are certainly not limited to, quantum state preparation, quantum error
correction, quantum metrology, and fundamental tests of quantum
mechanics \cite{wiseman_milburn,smooth,hypothesis,testing_quantum}.

From a decision-theoretic point of view, the Bayesian theory is
optimal only if the model perfectly matches the reality. In practical
situations, however, assumptions and approximations must be made, and
the mismatch between the model and the reality will introduce excess
systematic errors. General theoretical results concerning mismatched
estimation are highly desirable for practical filter design purposes
but difficult to obtain, especially if the dynamics is nonlinear.  In
this regard, a few interesting identities that relate mismatched
estimation to relative entropy for classical Gaussian or Poissonian
channels have recently been discovered \cite{verdu,weissman,atar,no},
building upon earlier seminal work that relates estimation theory to
Shannon mutual information \cite{duncan_mi,guo,guo08}. These relations
open up novel research directions and have already proved useful for
deriving a variety of new results, as they enable a fresh attack on
estimation problems using information-theoretic tools, and vice versa.

In this paper, I generalize two of these identities to the quantum
regime and relate mismatched quantum filtering errors to relative
entropy for continuous Gaussian or Poissonian measurements.
Given the plethora of new results that have since been spawned from
the classical relations, the quantum relations are envisioned to be
similarly fruitful in both quantum estimation theory and quantum
information theory and ultimately useful for quantum filter design.

\section{\label{filter}Mismatched Quantum Filtering}
For Gaussian measurements, such as homodyne detection of an optical
beam interacting with a quantum system, define
\begin{align}
Y_t \equiv \BK{y_\tau,0\le \tau \le t}
\end{align}
as the observation record up to time $t$. The posterior statistics of
the quantum system can be determined from the linear Belavkin equation
\cite{wiseman_milburn,belavkin05,bouten}, a quantum generalization of
the Duncan-Mortensen-Zakai equation:
\begin{align}
df_t(Y_t) &\equiv f_{t+dt}(Y_{t+dt})-f_t(Y_t) \nonumber\\
&= \mathcal L_t f_t(Y_t)dt + 
\frac{1}{2}\Bk{a_t f_t(Y_t) + f_t(Y_t)a_t^\dagger}dy_t,
\label{belavkin}
\end{align}
where $f_t(Y_t)$ is the unnormalized posterior density operator in the
Hilbert space for the quantum system, $a_t$ is an operator that
characterizes the interaction between the system and the probe, such
that
\begin{align}
q_t &\equiv \frac{1}{2} \bk{a_t+a_t^\dagger}
\end{align}
is the system observable being measured, $\mathcal L_t$ is a Lindblad
superoperator that describes the system dynamics, including the effect
of measurement backaction as a function of $a_t$, and $dy_t$ is the
increment of the observation process defined as $dy_t \equiv y_{t+dt}
- y_t$, with $dy_t^2 = dt$. The initial condition is given by the
initial density operator $\rho_0$:
\begin{align}
f_0 = \rho_0.
\end{align}
Measurement-based feedback control can be modeled by making
$(a_t,\mathcal L_t)$ depend on $Y_t$. 

The expectation of a function $g(O_t,Y_t)$ in terms of any observable
$O_t$ is given by
\begin{align}
\expect g(O_t,Y_t) &= \int dP_0(Y_t) \trace f_t(Y_t) g(O_t,Y_t),
\label{estimate}
\end{align}
where $dP_0(Y_t)$ is the probability measure for the standard Wiener
process.  The probability measure of an observation record is thus
\begin{align}
dP(Y_t) &= \expect \bs 1_{Y_t} = dP_0(Y_t) \trace f_t(Y_t),
\label{wiener}
\end{align}
where $\bs 1_{Y_t}$ is the indicator function.  The conditional
expectation of an observable $O_t$ is given by
\begin{align}
\expect\bk{O_t|Y_t} &= \frac{\expect(O_t \bs 1_{Y_t})}{\expect \bs 1_{Y_t}}
= \frac{\trace f_t(Y_t) O_t}{\trace f_t(Y_t)}
= \trace \rho_t(Y_t) O_t,
\label{condexp}
\end{align}
with the normalized posterior density operator given by
\begin{align}
\rho_t(Y_t)  &= \frac{f_t(Y_t)}{\trace f_t(Y_t)}.
\end{align}

Define a filtering estimator of the observable $q_t$ as $\check
q_t(Y_t)$.  A common measure of the filtering error is
\begin{align}
\textrm{cmse}_t &\equiv \expect \Bk{q_t-\check q_t(Y_t)}^2,
\label{mse}
\end{align}
where cmse is short for causal mean-square error. It is not difficult
to show that the quantum conditional expectation of $q_t$ minimizes
$\textrm{cmse}_t$ \cite{bouten}, analogous to the classical case:
\begin{align}
  \textrm{cmmse}_t &\equiv \inf_{\check q_t(Y_t)}\textrm{cmse}_t
  = \expect \Bk{q_t-\expect\bk{q_t|Y_t}}^2.
\label{mmse}
\end{align}
The amount of error in excess of the minimum value is called regret in
decision theory \cite{berger}. For mismatched quantum filtering with
Gaussian measurements, I define a regret quantity as the excess
mean-square error integrated over time:
\begin{align}
\Pi \equiv 
\frac{1}{2}\int_{0}^{T} dt \bk{\textrm{cmse}_t - \textrm{cmmse}_t},
\label{regret}
\end{align}
where the factor of $1/2$ is for later technical convenience.

\section{\label{hypothesis}Quantum Hypothesis Testing}
Consider now a different statistical problem: the discrimination of
two quantum models via continuous Gaussian measurements. A central
quantity in this binary hypothesis testing problem is the likelihood
ratio, defined as
\begin{align}
\Lambda(Y_T) &\equiv \frac{dP(Y_T)}{dP'(Y_T)},
\end{align}
where $dP(Y_T)$ is the probability measure of $Y_T$ assuming the first
model and the prime denotes the same quantity but assuming the second
model. Eq.~(\ref{wiener}) enables one to relate $\Lambda(Y_t)$ to the
quantum filters as
\begin{align}
\Lambda(Y_T) &= \frac{\trace f_T(Y_T)}{\trace f_T'(Y_T)},
\label{lr}
\end{align}
where $f'$ obeys another linear Belavkin equation that assumes the
second model:
\begin{align}
df_t'(Y_t) 
&= \mathcal L_t' f_t'(Y_t)dt + 
\frac{1}{2}\Bk{a_t' f_t'(Y_t) + f_t'(Y_t)a_t'^\dagger}dy_t,
\label{wrong_belavkin}
\end{align}
with the measured observable defined as
\begin{align}
q_t' &\equiv \frac{1}{2}\bk{a_t'+a_t'^\dagger},
\end{align}
and the initial condition given by
\begin{align}
f_0' = \rho_0'.
\end{align}
The conditional expectation assuming the second model becomes
\begin{align}
\expect'\bk{O_t'|Y_t} &= \frac{\trace f_t'(Y_t)O_t'}{\trace f_t'(Y_t)}
= \trace \rho_t'(Y_t)O_t'.
\end{align}
The following identity, first derived in Ref.~\cite{hypothesis} and
generalizing a classical result by Duncan \cite{duncan_lr}, will be
useful:
\begin{lemma}\label{lem_duncan}
  The log-likelihood ratio for two quantum models under continuous
  Gaussian measurements satisfies
\begin{align}
\ln \Lambda(Y_T) &= \int_{0}^T dy_t \Bk{\expect\bk{q_t|Y_t}-\expect'\bk{q_t'|Y_t}}
\nonumber\\&\quad
-\frac{1}{2}\int_{0}^T dt 
\Bk{\expect^2\bk{q_t|Y_t}-\expect'^2\bk{q_t'|Y_t}},
\label{duncan}
\end{align}
where $\expect\bk{q_t|Y_t}$ and $\expect'\bk{q_t'|Y_t}$ are the
filtering conditional expectations of the measured observable under
the two models.
\end{lemma}
\begin{proof}
  Tracing over Eqs.~(\ref{belavkin}) and (\ref{wrong_belavkin}), one
  obtains $\trace df_t = d\trace f_t = \expect\bk{q_t|Y_t} dy_t \trace
  f_t$ and $d\trace f_t' = \expect'\bk{q_t'|Y_t} dy_t \trace f_t'$, as
  the trace of a Lindblad superoperator on any operator is zero.
  It\={o} calculus can then be used to compute $d \ln \trace f_t =
  d\trace f_t/\trace f_t -(d\trace f_t/\trace f_t)^2/2 =
  \expect\bk{q_t|Y_t} dy_t-\expect^2\bk{q_t|Y_t}dt/2$, where the last
  step uses $dy_t^2 = dt$. A similar formula can be derived for $d\ln
  \trace f_t'$. Integrating $d\ln \trace f_t$ and $d\ln \trace f_t'$
  over time and plugging them into Eq.~(\ref{lr}) results in
  Eq.~(\ref{duncan}).
\end{proof}
The relative entropy between the two probability measures is defined
as the expectation of the log-likelihood ratio $\ln \Lambda(Y_T)$
assuming the first model:
\begin{align}
D(dP||dP') &\equiv \expect \ln \Lambda(Y_T),
\label{kl}
\end{align}
which is a well known information quantity relevant to many
statistical applications \cite{cover}.


\section{\label{formula} Filter Regret
and Relative Entropy}
The first main result of this paper is the following theorem,
generalizing a classical result by Weissman \cite{weissman}:
\begin{theorem}\label{thm_gauss}
  For continuous Gaussian measurements, the regret for mismatched
  quantum filtering is equal to the relative entropy between the 
  true and nominal observation probability measures; viz.,
\begin{align}
\Pi &= D(dP||dP').
\end{align}
\end{theorem}
\begin{proof}
  Substituting Eq.~(\ref{duncan}) in Lemma~\ref{lem_duncan} into
  Eq.~(\ref{kl}) and interchanging the order of integration and
  expectation,
\begin{align}
D(dP||dP') &= 
\int_{0}^T \expect\BK{dy_t \Bk{\expect\bk{q_t|Y_t}-\expect'\bk{q_t'|Y_t}}}
\nonumber\\&\quad
-\frac{1}{2}\int_{0}^T dt \expect\Bk{\expect^2\bk{q_t|Y_t}-\expect'^2\bk{q_t'|Y_t}}.
\end{align}
For the first expectation, one can use the orthogonality principle of
the conditional expectation to write
\begin{align}
&\quad 
\expect\BK{dy_t \Bk{\expect\bk{q_t|Y_t}-\expect'\bk{q_t'|Y_t}}}
\nonumber\\
&= 
\expect\BK{\expect\bk{dy_t|Y_t} \Bk{\expect\bk{q_t|Y_t}-\expect'\bk{q_t'|Y_t}}}
\\
&= \expect\BK{\expect\bk{q_t|Y_t}dt \Bk{\expect\bk{q_t|Y_t}-\expect'\bk{q_t'|Y_t}}},
\end{align}
where the second step follows from the martingale property of the
quantum innovation process $\expect\Bk{dy_t-\expect\bk{q_t|Y_t}dt|Y_t}
= 0$ \cite{belavkin,wiseman_milburn,belavkin05,bouten}.  This results
in \cite{vanhandel_thesis}
\begin{align}
D(dP||dP') &= \frac{1}{2} \int_{0}^T dt 
\expect \Bk{\expect\bk{q_t|Y_t}-\expect'\bk{q_t'|Y_t}}^2.
\end{align}
The regret given by Eq.~(\ref{regret}), on the other hand, is
\begin{align}
\Pi &= \frac{1}{2}
\int_{0}^T dt \expect \BK{\Bk{q_t-\expect'\bk{q_t'|Y_t}}^2
-\Bk{q_t-\expect\bk{q_t|Y_t}}^2}
\\
&= \frac{1}{2}
\int_{0}^T dt \expect 
\left\{
2q_t\Bk{\expect\bk{q_t|Y_t}-\expect'\bk{q_t'|Y_t}}
\right.\nonumber\\&\quad\left.
+\expect'^2\bk{q_t'|Y_t}-\expect^2\bk{q_t|Y_t}
\right\}
\\
&= \frac{1}{2}
\int_{0}^T dt \expect \left\{
2\expect\bk{q_t|Y_t}\Bk{\expect\bk{q_t|Y_t}-\expect'\bk{q_t'|Y_t}}
\right.\nonumber\\&\quad\left.
+\expect'^2\bk{q_t'|Y_t}-\expect^2\bk{q_t|Y_t}
\right\}
\label{step}
\\
&= \frac{1}{2}
\int_{0}^T dt \expect 
\Bk{\expect\bk{q_t|Y_t}-\expect'\bk{q_t'|Y_t}}^2
= D(dP||dP'),
\end{align}
where Eq.~(\ref{step}) uses the orthogonality principle for the
quantum conditional expectation 
$\expect\BK{\Bk{q_t-\expect\bk{q_t|Y_t}}g(Y_t)} = 0$ \cite{bouten},
which is valid for any $g(Y_t)$.
\end{proof}
Apart from the assumption of continuous Gaussian measurements,
Theorem~\ref{thm_gauss} is applicable to arbitrary time $T$ and rather
general quantum Markov models, which shall hereafter be denoted by
\begin{align}
\mathcal M &\equiv \BK{\rho_0,a_t,\mathcal L_t; 0\le t \le T},
\\
\mathcal M' &\equiv \BK{\rho_0',a_t',\mathcal L_t'; 0\le t \le T}.
\end{align}
The theorem is also applicable to adaptive models, if one makes
$(a_t,\mathcal L_t)$ and/or $(a_t',\mathcal L_t')$ depend on $Y_t$.

\section{\label{app}Implications}

\subsection{Bayes Quantum Filtering and Mutual Information}
Suppose that the model $\mathcal M$ is chosen from an ensemble
$\BK{d\pi(\theta),\mathcal M_\theta}$ parametrized by
$\theta$. The prior probability measure for $\theta$ is defined as
$d\pi(\theta)$, the expectation under which is denoted by
$\expect_\theta$. Assume that the true model has access to the exact
$\theta$, or $\mathcal M = \mathcal M_\theta$, such that
$\expect\bk{q_t|Y_t} = \expect\bk{q_t|Y_t,\theta}$ and $dP(Y_T) =
dP_\theta(Y_T)$, but the nominal model does
not. Theorem~\ref{thm_gauss} can then be used to relate the expected
regret for not knowing $\theta$ to the cross information:
\begin{align}
\expect_\theta\Pi &= \expect_\theta D(dP_\theta||dP').
\label{cross}
\end{align}
If the nominal model has access to $d\pi(\theta)$, the optimal filter
should be a Bayes estimator, and $\inf_{\mathcal M'}
\expect_\theta\Pi$ is the Bayes regret.  This turns out to be equal
to the mutual information:
\begin{cor}\label{mi}
  The Bayes ignorance regret is equal to the mutual information; viz.,
\begin{align}
\inf_{\mathcal M'} \expect_\theta\Pi 
= I(\theta;Y) \equiv \expect_\theta D(dP_\theta|| \expect_\theta dP_\theta).
\label{mi_eq}
\end{align}
\end{cor}
\begin{proof}
  The Bayes filter that minimizes $\expect_\theta \textrm{cmse}_t$ and
  therefore $\expect_\theta\Pi$ is $\expect'\bk{q_t'|Y_t} =
  \expect_\theta\Bk{\expect\bk{q_t|Y_t,\theta}|Y_t}$, with $dP' =
  \expect_\theta dP_\theta$. Substituting $dP' = \expect_\theta
  dP_\theta$ into Eq.~(\ref{cross}) results in Eq.~(\ref{mi_eq}).
\end{proof}
\begin{remark}
  The classical relation between mutual information and filtering
  error \cite{duncan_mi} can be derived from Corollary~\ref{mi} by
  setting $\theta = \{q_\tau,0\le \tau \le T\}$ and noting that
  $\expect(q_t|Y_t,\theta) = q_t$ and $\textrm{cmmse}_t = 0$. In the
  quantum case, the history of an observable has questionable
  decision-theoretic meaning unless it is a quantum nondemolition
  observable \cite{yanagisawa}, but the more general
  Corollary~\ref{mi} still holds.
\end{remark}

Corollary~\ref{mi} gives a new operational meaning to mutual
information as a measure of parameter importance in quantum filtering:
high $I(\theta;Y)$ means more regret for not knowing $\theta$ and
$\theta$ is thus worth knowing in the context of filtering, while low
$I(\theta;Y)$ means less regret for ignoring $\theta$.

\subsection{Minimax Quantum Filtering and Channel Capacity}
If the prior $d\pi(\theta)$ is not known except that it belongs to a
certain set, one can consider the maximin regret $\sup_{d\pi}
\inf_{\mathcal M'} \expect_\theta\Pi$, which is the worst possible
Bayes regret. This is related to the channel capacity as a direct
result of Corollary~\ref{mi}:
\begin{cor}\label{capacity}
The maximin ignorance regret is equal to the channel capacity; viz.,
\begin{align}
\sup_{d\pi} \inf_{\mathcal M'}\expect_\theta\Pi &= C \equiv
\sup_{d\pi}  I(\theta;Y),
\end{align}
and the least-favorable prior is equal to the capacity-attaining
prior; viz.,
\begin{align}
\arg \sup_{d\pi} \inf_{\mathcal M'}\expect_\theta\Pi &= 
d\pi^*(\theta) \equiv \arg \sup_{d\pi} I(\theta;Y).
\end{align}
\end{cor}
Consider also the minimax regret $\inf_{\mathcal M'} \sup_{d\pi}
\expect_\theta\Pi$, which uses a minimax filter that minimizes the
worst possible regret. The channel-capacity connection can be
exploited to prove the following, similar to the classical result
\cite{no}:
\begin{cor}\label{minimax}
  The minimax and maximin ignorance regrets are equal and
  given by the channel capacity; viz.,
\begin{align}
\inf_{\mathcal M'} \sup_{d\pi} \expect_\theta\Pi &= 
\sup_{d\pi} \inf_{\mathcal M'}\expect_\theta\Pi = C,
\label{minimax_eq}
\end{align}
and the minimax filter is equivalent to the Bayes filter with the
least-favorable prior $d\pi^*(\theta)$.
\end{cor}
\begin{proof}
  The proof may be done by applying the minimax theorem \cite{berger}
  to quantum filtering, but here I shall use information
  theory instead. Let
$d\pi^*(\theta) \equiv \arg \sup_{d\pi} I(\theta;Y)$
be the capacity-attaining prior, the expectation under which is
denoted by $\expect^*_\theta$.  The redundancy-capacity theorem states
that \cite{cover,haussler}
\begin{align}
C &= \inf_{dP'} \sup_{d\pi} \expect_\theta D(dP_\theta||dP'),
\label{redundancy}
\end{align}
and the minimax $dP'$ is
\begin{align}
dP'^* \equiv \arg \inf_{dP'} \sup_{d\pi} \expect_\theta D(dP_\theta||dP')
&= \expect_\theta^* dP_\theta.
\end{align}
A Bayes filter with model $\mathcal M'^*$ and $dP'^*=\expect_\theta^*
dP_\theta$ exists, so $\inf_{dP'}$ in Eq.~(\ref{redundancy}) can be
replaced with $\inf_{\mathcal M'}$:
\begin{align}
C &= \inf_{\mathcal M'} \sup_{d\pi} \expect_\theta D(dP_\theta||dP')
=\inf_{\mathcal M'} \sup_{d\pi} \expect_\theta\Pi,
\end{align}
where the last step uses Theorem~\ref{thm_gauss}. Combining this with
Corollary~\ref{capacity} leads to Corollary~\ref{minimax}.
\end{proof}

\subsection{Quantum Information Bounds}
Perhaps the most remarkable property of Theorem~\ref{thm_gauss} is
that it relates the regret for mismatched quantum filtering to the
amount of information for binary hypothesis testing, such that a
limitation on one application implies a guaranteed performance for the
other.  Upper bounds on the filter regrets should be particularly
useful for robust quantum estimation and control design
\cite{james04,stockton,yamamoto09,szigeti} and proving the stability
of quantum filters \cite{vanhandel_thesis,vanhandel09}.

For example, suppose that that the two models share identical dynamics and
measurements and differ only in the initial conditions $\rho_0$ and
$\rho_0'$.  The observation probability measures can then be expressed
with respect to the same positive operator-valued measure (POVM)
$d\mu(Y_T)$ \cite{wiseman_milburn}:
\begin{align}
dP(Y_T) &= \trace \rho_0 d\mu(Y_T),
&
dP'(Y_T) &= \trace \rho_0' d\mu(Y_T),
\end{align}
and quantum upper bounds on the regrets can be obtained as follows:
\begin{cor}\label{qre}
  If the two models differ only in the initial conditions, the filter
  regret is bounded by the quantum relative entropy between the two
  initial density operators; viz.,
\begin{align}
\Pi &\le D(\rho_0||\rho_0') \equiv \trace \rho_0 \bk{\ln \rho_0-\ln \rho_0'}.
\end{align}
\end{cor}
\begin{proof}
  $\Pi = D(dP||dP')$ from Theorem~\ref{thm_gauss}, and
  it is known from quantum information theory that $D(dP||dP')\le
  D(\rho_0||\rho_0')$ for any $d\mu(Y_T)$ \cite{wilde13}.
\end{proof}
Corollary~\ref{qre} proves that the time-averaged regret $\Pi/T$ due
to a mismatched initial condition is guaranteed to decrease inversely
with time if $D(\rho_0||\rho_0')<\infty$.

Regrets due to ignorance can also be bounded by quantum information
quantities as follows:

\begin{cor}\label{holevo}
  If $dP(Y_T) = \trace \rho_{\theta} d\mu(Y_T)$ and $dP'(Y_T) = \trace
  \rho' d\mu(Y_T)$, the Bayes and minimax ignorance regrets are
  bounded by the Holevo information $\chi$; viz.,
\begin{align}
\inf_{\mathcal M'}
\expect_\theta \Pi &\le \chi\BK{d\pi(\theta),\rho_{\theta}}
\equiv \expect_\theta D(\rho_{\theta}||\expect_\theta  \rho_{\theta}),
\label{holevo_bound}\\
\inf_{\mathcal M'}\sup_{d\pi}
\expect_\theta \Pi &\le \sup_{d\pi} \chi\BK{d\pi(\theta),\rho_{\theta}}.
\label{minimax_holevo}
\end{align}
\end{cor}
\begin{proof}
  $\inf_{\mathcal M'}\expect_\theta\Pi = I(\theta,Y)$ from
  Corollary~\ref{mi} and the Holevo bound states that $I(\theta;Y) \le
  \chi$ for any POVM
  \cite{wilde13}. Eq.~(\ref{minimax_holevo}) follows
  from Corollary~\ref{minimax} and Eq.~(\ref{holevo_bound}).
\end{proof}

\section{\label{poisson}Poissonian Measurements}
The quantum filter for continuous Poissonian measurements, such as
photon counting of the optical probe beam, is similar to the Gaussian
case, except that the unnormalized filtering equation now reads
\cite{wiseman_milburn,bouten}
\begin{align}
df_t(Y_t) &= \mathcal L_t f_t(Y_t) dt +\Bk{a_t f_t(Y_t)a_t^\dagger -f_t(Y_t)}(dy_t-dt),
\end{align}
and the measured observable is now $q_t \equiv a_t^\dagger a_t$.  It
is not difficult to show that $dP(Y_t) = dP_0(Y_t)\trace f_t(Y_t)$,
where $dP_0(Y_t)$ is the probability measure of a standard Poisson
process with $\expect_0(dy_\tau|Y_\tau) = dt$.  The log-likelihood
ratio satisfies the following identity,
similar to Lemma~\ref{lem_duncan} (see the Supplemental Material of
Ref.~\cite{hypothesis} for a proof):
\begin{lemma}\label{lem_snyder}
The log-likelihood ratio for two quantum models under
continuous Poissonian measurements satisfies
\begin{align}
\ln \Lambda(Y_T) &= \int_{0}^T dy_t \ln \frac{\expect\bk{q_t|Y_t}}{\expect'\bk{q_t'|Y_t}}
\nonumber\\&\quad
-\int_{0}^T dt \Bk{\expect\bk{q_t|Y_t}-\expect'\bk{q_t'|Y_t}}.
\label{snyder}
\end{align}
\end{lemma}
To obtain a result analogous to Theorem~\ref{thm_gauss}, I follow
Atar and Weissman \cite{atar} and define the following loss function
instead of the quadratic criterion:
\begin{align}
l(q,\check q) &\equiv q \ln \frac{q}{\check q} - q + \check q.
\label{loss}
\end{align}
The mean-loss error of a causal estimate $\check q_t(Y_t)$ at time
$t$ becomes
\begin{align}
\textrm{cmle}_t &\equiv \expect l\bk{q_t,\check q_t(Y_t)}.
\end{align}
It is easy to show that the conditional expectation $\expect\bk{q_t|Y_t}$
minimizes this error as well, such that the minimum error is
\begin{align}
\textrm{cmmle}_t &\equiv \expect l\bk{q_t,\expect\bk{q_t|Y_t}},
\end{align}
and the regret can then be defined as
\begin{align}
\Pi_l \equiv \int_{0}^T dt \bk{\textrm{cmle}_t-\textrm{cmmle}_t}.
\label{regret_poisson}
\end{align}
The second main result of this paper thus follows naturally as a
generalization of a classical result by Atar and Weissman \cite{atar}:
\begin{theorem}\label{thm_poisson}
  For continuous Poissonian measurements,
\begin{align}
\Pi_l &= D(dP||dP').
\end{align}
\end{theorem}
\begin{proof}
  The proof is similar to that for Theorem~\ref{thm_gauss}.  Taking
  the expectation of Eq.~(\ref{snyder}) in Lemma~\ref{lem_snyder} and
  noting the martingale property for Poissonian measurements given by
  $\expect\Bk{dy_t-\expect\bk{q_t|Y_t} dt|Y_t} = 0$
  \cite{barchielli_holevo}, the relative entropy can be written as
\begin{align}
D(dP||dP') = \int_{0}^T dt 
\expect l\bk{\expect\bk{q_t|Y_t},\expect'\bk{q_t'|Y_t}}.
\end{align}
The regret, on the other hand, is
\begin{align}
\Pi_l 
&= \int_{0}^T dt \expect 
\left\{
\Bk{q_t \ln \frac{q_t}{\expect'\bk{q_t'|Y_t}}-
q_t+\expect'\bk{q_t'|Y_t}}
\right.\nonumber\\&\quad\left.
-\Bk{q_t \ln \frac{q_t}{\expect\bk{q_t|Y_t}}-
q_t+\expect\bk{q_t|Y_t}}\right\}
\\
&= 
\int_{0}^T dt \expect \Bk{q_t \ln \frac{\expect\bk{q_t|Y_t}}{\expect'\bk{q_t'|Y_t}}
-\expect\bk{q_t|Y_t}+\expect'\bk{q_t'|Y_t}}
\\
&= 
\int_{0}^T dt \expect
\left[\expect\bk{q_t|Y_t} \ln \frac{\expect\bk{q_t|Y_t}}{\expect'\bk{q_t'|Y_t}}
-\expect\bk{q_t|Y_t}\right.
\nonumber\\&\quad
+\expect'\bk{q_t'|Y_t}\bigg]
\label{orth2}
\\
&=
\int_{0}^T dt 
\expect l\bk{\expect\bk{q_t|Y_t},\expect'\bk{q_t'|Y_t}}
=D(dP||dP'),
\end{align}
where Eq.~(\ref{orth2}) again follows from the orthogonality principle
for quantum conditional expectations.
\end{proof}
One direct consequence of Theorem~\ref{thm_poisson} is that
Corollaries~\ref{mi}--\ref{holevo} are also applicable to Poissonian
measurements, if we consider $\Pi_l$ instead of $\Pi$.

\section{Conclusion}
With Theorems~\ref{thm_gauss} and \ref{thm_poisson}, I have taken the
first step towards a quantum generalization of the fascinating
connections between estimation theory and Shannon information theory
for Gaussian and Poissonian channels.  The presented results are
envisioned to aid the study of quantum estimation and control
techniques for complex systems, as they enable one to analyze and
design quantum filters using techniques borrowed from information
theory.  Regardless of the potential applications, these new relations
between central quantities in quantum estimation and information
theory are bound to bring fresh insights to both areas.

\section*{Acknowledgment}
This work is supported by the Singapore National Research Foundation
under NRF Grant No.~NRF-NRFF2011-07.


\bibliographystyle{IEEEtran}
\bibliography{ISIT2014}

\end{document}